\documentclass[11pt]{article}

\usepackage[utf8]{inputenc}
\usepackage{amsmath,amsthm,amssymb,caption,subcaption,nicefrac}
\usepackage{fullpage}
\usepackage{bm}
\usepackage[table]{xcolor}
\usepackage{multirow}
\usepackage{arydshln}
\usepackage{csquotes}
\usepackage{anyfontsize}
\usepackage[shortlabels]{enumitem}

\usepackage{amsmath}
\usepackage{amssymb}
\usepackage{mathtools}
\usepackage{amsthm}

\usepackage[T1]{fontenc} 

\usepackage{algorithm}
\usepackage[noend]{algorithmic}
\usepackage[algo2e,ruled,noend]{algorithm2e}

\usepackage[most]{tcolorbox}
\usepackage{lipsum}
\usepackage[hybrid]{markdown}

\usepackage{hyperref}
\hypersetup{
   colorlinks=true,
   linkcolor=olive,
   citecolor=purple
}

\usepackage{cleveref}
\usepackage{aliascnt}

\definecolor{Gred}{RGB}{219, 50, 54}
\definecolor{Ggreen}{RGB}{60, 186, 84}
\definecolor{Gblue}{RGB}{72, 133, 237}
\definecolor{Gyellow}{RGB}{247, 178, 16}
\definecolor{ToCgreen}{RGB}{0, 128, 0}
\definecolor{myGold}{RGB}{231,141,20}
\definecolor{myBlue}{rgb}{0.19,0.41,.65}
\definecolor{myPurple}{RGB}{175,0,124}

\newtheorem{theorem}{Theorem}

\newaliascnt{proposition}{theorem}
\newtheorem{proposition}[proposition]{Proposition}
\aliascntresetthe{proposition}

\newaliascnt{lemma}{theorem}
\newtheorem{lemma}[lemma]{Lemma}
\aliascntresetthe{lemma}

\newaliascnt{corollary}{theorem}
\newtheorem{corollary}[corollary]{Corollary}
\aliascntresetthe{corollary}

\newaliascnt{fact}{theorem}

\aliascntresetthe{fact}

\newaliascnt{definition}{theorem}
\newtheorem{definition}[definition]{Definition}
\aliascntresetthe{definition}

\newaliascnt{observation}{theorem}

\aliascntresetthe{observation}

\newaliascnt{conjecture}{theorem}

\aliascntresetthe{conjecture}

\newaliascnt{assumption}{theorem}

\aliascntresetthe{assumption}

\newaliascnt{claim}{theorem}
\newtheorem{claim}[claim]{Claim}
\aliascntresetthe{claim}

\newaliascnt{remark}{theorem}

\aliascntresetthe{remark}

\crefname{theorem}{theorem}{theorems}
\Crefname{theorem}{Theorem}{Theorems}

\crefname{proposition}{proposition}{propositions}
\Crefname{proposition}{Proposition}{Propositions}

\crefname{lemma}{lemma}{lemmas}
\Crefname{lemma}{Lemma}{Lemmas}

\crefname{corollary}{corollary}{corollaries}
\Crefname{corollary}{Corollary}{Corollaries}

\crefname{fact}{fact}{facts}
\Crefname{fact}{Fact}{Facts}

\crefname{definition}{definition}{definitions}
\Crefname{definition}{Definition}{Definitions}

\crefname{observation}{observation}{observations}
\Crefname{observation}{Observation}{Observations}

\crefname{conjecture}{conjecture}{conjectures}
\Crefname{conjecture}{Conjecture}{Conjectures}

\crefname{assumption}{assumption}{assumptions}
\Crefname{assumption}{Assumption}{Assumptions}

\crefname{claim}{claim}{claims}
\Crefname{claim}{Claim}{Claims}

\crefname{remark}{remark}{remarks}
\Crefname{remark}{Remark}{Remarks}

\renewcommand{\phi}{\varphi}

\newcommand{\defeq}{\coloneq}
\DeclareMathOperator{\score}{score}
\newcommand{\LR}[1]{\mathrm{LR}(#1)}

\allowdisplaybreaks

\newcommand{\probstyle}[1]{\textsc{#1}\xspace}
\newcommand{\wthiele}{\probstyle{$w$-Thiele}}

\newcommand{\genthiele}{\probstyle{Generalized Thiele}}
\newcommand{\thiele}{\probstyle{Thiele}}

\newtcolorbox{prompt}{
    enhanced,
    breakable,
    colback=blue!3!white,    
    colframe=cyan!30!black,  
    coltitle=white,          
    fonttitle=\bfseries,     
    title={User Prompt},
    boxrule=1pt,
    arc=4pt,                 
    drop shadow,             
    toptitle=3pt,
    bottomtitle=3pt
}

\newtcolorbox{answer}{
    enhanced,
    breakable,
    colback=white,           
    colframe=gray!40!white,  
    coltitle=black,          
    fonttitle=\bfseries,     
    title={AI Response},
    boxrule=0.5pt,
    arc=4pt,
    left=4pt, right=4pt, top=4pt, bottom=4pt,
    before=\vspace{0.5em},   
    after=\vspace{2em},      
    toptitle=3pt,
    bottomtitle=3pt
}

\title{Polynomial-Time Algorithm for Thiele Voting Rules with\\
Voter Interval Preferences}
\author{
    Pasin Manurangsi\\
    Google Research\\
    {\small\texttt{pasin@google.com}}\\[-2pt]
    \and
    Krzysztof Sornat\\
    AGH University, Poland\\[-2pt]
    {\small\texttt{sornat@agh.edu.pl}}
}
\date{}

\begin{document}

\maketitle

\begin{abstract}
  We present a polynomial-time algorithm for computing an optimal committee of size $k$ under any given Thiele voting rule for elections on the Voter Interval domain (i.e., when voters can be ordered so that each candidate is approved by a consecutive voters).
  Our result extends to the Generalized Thiele rule, in which each voter has an individual weight (scoring) sequence.
  This resolves a 10-year-old open problem that was originally posed for Proportional Approval Voting and later extended to every Thiele rule (Elkind and Lackner, IJCAI 2015; Peters, AAAI 2018).

  Our main technical ingredient is a new structural result---a concavity theorem for families of intervals.
  It shows that, given two solutions of different sizes, one can construct a solution of any intermediate size whose score is at least the corresponding linear interpolation of the two scores.
  As a consequence, on Voter Interval profiles, the optimal total Thiele score is a concave function of the committee size.
  We exploit this concavity within an optimization framework based on a Lagrangian relaxation of a natural integer linear program formulation, obtained by moving the cardinality constraint into the objective.
  On Voter Interval profiles, the resulting constraint matrix is totally unimodular, so it can be solved in polynomial time.

  Our main algorithm and its proof were obtained via human--AI collaboration.
  In particular, a slightly simplified version of the main structural theorem used by the algorithm was obtained in a single call to Gemini Deep Think.
\end{abstract}

\section{Introduction}
Approval-based committee (ABC) elections are an important topic in computational social choice literature~\cite{lackner23abc_book} having applications beyond committee selection, e.g., in recommendation systems~\cite{SkowronFL16_aij_set_of_items} and blockchain~\cite{BoehmerBCG0S24_polkadot}.
In ABC elections, voters cast votes in a form of approval ballots,
i.e., every voter indicates a subset of candidates.
An ABC voting rule takes a collection of approval ballots (an approval profile) and outputs a committee of a given size $k$.

One of a general class of ABC voting rules are Thiele rules~\cite{Thiele95} which contains, as special cases, such basic voting rules as:
Multiwinner Approval Voting (AV),
Proportional Approval Voting (PAV)~\cite{AzizGGMMW15}
and Approval Voting Chamberlin-Courant (CC)~\cite{ChamberlinCourant83},
which covers different properties from efficiency via proportionality to excellence~\cite{LacknerS21}.
A Thiele voting rule is parameterized by an a non-increasing infinite sequence of weights $(w_1,w_2,\dots)$. The partial sum $w_1+w_2+\dots+w_R$ models satisfaction of a voter (score of a voter) from receiving $R$ representatives in a committee.
A $w$-Thiele rule selects a committee of size $k$ that maximizes total score, i.e., sum of voters scores.

From the computational complexity point of view,
selecting an optimal committee under Thiele rules is NP-hard (except for AV)~\cite{AzizGGMMW15,SkowronFL16_aij_set_of_items}.
Therefore, computability of Thiele rules has been investigated under different paradigms:
approximation algorithms~\cite{SkowronFL16_aij_set_of_items,ByrkaSS18_hkm,DudyczMMS20_ijcai,BarmanFF21_concave_coverage,BarmanFGG22_lcoverage},
parameterized algorithms~\cite{BetzlerSU13_jair,ElkindL15_ijcai15,YangW18-aamas18,BredereckF0KN20_aaai,SornatWX22_ijcai,YangW23-jaamas,GuptaJSSU2025_independent_fpt,LassotaSornat_aaai26}
and structured domains~\cite{BetzlerSU13_jair,ElkindL15_ijcai15,ElkindLP2017trends_book,ConstantinescuE21_sc_linear,SornatWX22_ijcai,ElkindLP25_arxiv_preference_restrictions,LassotaSornat_aaai26}.
In particular, Thiele rules on Candidate Interval (CI) domain (i.e., when candidates can be ordered so that each voter approves consecutive candidates)
are polynomial-time solvable by a proper Integer Linear Programming (ILP) formulation which has been proven to be totally unimodular (TU)~\cite{Peters18_aaai}.
This result extends to dichotomous Single-Peaked on a Circle profiles (this domain generalizes approving intervals of candidates on a path into approving intervals of candidates on a circle graph)~\cite{PetersL20_spoc_jair}.
An analogous to CI is the Voter Interval (VI) domain which contains all approval profiles with the property that voters can be ordered so that each candidate is approved by a consecutive voters.

Despite the effort made over the years and tackling Thiele rules under different paradigms and using different techniques,
the computational complexity of solving any given Thiele rule on a VI profile is already a 10-year-old open problem~\cite{ElkindL15_ijcai15,ElkindLP2017trends_book,Peters18_aaai,YangW18-aamas18,PetersL20_spoc_jair,GodziszewskiB0F21_vci_aaai21,lackner23abc_book,YangW23-jaamas,ElkindLP25_arxiv_preference_restrictions,LassotaSornat_aaai26}.
It did not help that VI profiles are analogous to CI profiles by transposing voters with candidates.
In particular, as observed by Peters~\cite{Peters18_aaai}, a corresponding ILP formulation for VI profiles is not TU, in contrast to ILP formulation for CI profiles.
This comes from the fact that the constraint matrix for VI has consecutive-ones property column-wise,
but the cardinality constraint correspond to a row of ones.
A concrete family of non-TU instances (with unbounded determinant prohibiting TU property) has been given in~\cite{LassotaSornat_aaai26}.
Although, the constraint matrix not being TU does not imply hardness of the problem.

Over the time, partial progress has been established with two approaches.
First, for some particular Thiele rules on VI there have been designed polynomial time algorithms.
Namely, for CC~\cite{BetzlerSU13_jair} and for all Thiele rules that are $\ell$-truncated, i.e., when $w_i = 0$ for $i \ge \ell+1$~\cite{ElkindL15_ijcai15}.
Without knowing the problem is NP-hard,
second approach was designing FPT algorithms, e.g., algorithms parameterized by the size of an approval ballot and by the support size~\cite{ElkindL15_ijcai15,LassotaSornat_aaai26}.

\subsection{Our Contribution}
We resolve this open problem by giving a polynomial-time algorithm that, for any fixed Thiele voting rule, computes an optimal size-$k$ committee for elections on the Voter Interval domain.
Our result also extends to the Generalized Thiele rule, in which each voter has an individual weight (scoring) sequence.

Our main technical ingredient is a novel structural result---a concavity theorem for VI profiles (see \Cref{sec:concavity}): given any two committees $S,S'$ (not necessarily optimal) of sizes $k_1<k_2$, we can construct, for any $k\in[k_1,k_2]$, a size-$k$ committee contained in $S\cup S'$ whose total score is at least the corresponding linear interpolation of the scores of $S$ and $S'$.
By recursively applying this construction, we can obtain committees of all intermediate sizes in polynomial time.
In particular, it implies that the optimal total score is a concave function of the committee size on VI profiles.
  
To compute an optimal size-$k$ committee (see \Cref{sec:poly-time-alg}), we use a Lagrangian relaxation of a natural ILP formulation, moving the cardinality constraint into the objective.
On VI profiles, the remaining constraint matrix becomes TU, so for any Lagrange multiplier we can compute an integral optimum in polynomial time.
We then perform a binary search over the multiplier values to obtain solutions bracketing size $k$---either it finds a solution of size exactly $k$ or it finds two committees: one of size at most $k$ and one of size at least $k$.
In the latter case, using the concavity theorem, we convert such bracketing solutions into an optimal committee of size exactly $k$.

Finally, we remark that key ideas for our main proofs are obtained via human--AI collaboration through iterative interactions with Gemini.
In particular, both the main idea of using Lagrangian relaxation and the main structural result (concavity theorem) are generated by Gemini---which later were adapted, fully verified and slightly generalized by the authors.
We view this as a concrete case study of AI-assisted research in theoretical computer science open problems.
A detailed account of our interaction with the model is given in \Cref{sec:ai-assistance}.

\section{Preliminaries}

We use the notation $[n] = \{1,2,\dots,n\}$.
Let $C = \{c_1, \dots, c_m\}$ be a set of candidates and $V = \{v_1, \dots, v_n\}$ be a set of voters.
$A_v \subseteq C$ denotes an \emph{approval ballot} of voter $v \in V$.
A collection of approval ballots $A = (A_v)_{v \in V}$ is called an \emph{approval profile}.
Any subset $W \subseteq C$ is called a \emph{committee}.
The \emph{set of supporters} of $c \in C$ is denoted by $V_c \defeq \{ v \in V: c \in A_v \}$.
An \emph{approval-based committee} (ABC) election is usually defined as a tuple $E = (C,V,A,k)$,
where $k \in \mathbb{N}$ is the required committee size.
In this paper, however, we represent an election simply by a pair $(A,k)$,
since from $(A,k)$ we can infer the set of voters and the set of candidates $C = \bigcup_{v \in V} A_v$, i.e.,
the candidates that receive at least one approval.
For the computational problems studied in this paper,
candidates with zero approvals are relevant only in the trivial case $\left| \bigcup_{v \in V} A_v \right| < k$.
An \emph{ABC voting rule} is a function that maps an election to a set of \emph{winning committees} of size $k$.

In this paper, we focus on $w$-Thiele voting rules,
which is a class of ABC voting rules that are parameterized by an infinite non-increasing sequence $w = (w_1, w_2, \dots)$, called a \emph{Thiele sequence}.
Winning committees under $w$-Thiele rules are those that achieve the maximum \emph{total score}:
\[ \max_{W \subseteq C, |W|=k} \score_w(W), \quad\textrm{where } \score_w(W) \defeq \sum_{v \in V} \sum_{i=1}^{|A_v \cap W|} w_i.\]

In the corresponding computational optimization problem \wthiele
one has to find a winning committee (of a given size).
By $\wthiele(A,k)$ we denote the maximum total score achieved by a winning committee of size $k$ on a given instance $(A,k)$ and we use a shortcut $\thiele(k)$ if $A$ and $w$ are clear from the context.

As is standard in the literature, we assume w.l.o.g. that $w_1 = 1$ as scaling the values of $w$ does not change the set of optimal solutions.
Multiwinner Approval Voting (AV) is a Thiele rule with $w = (1,1,\dots)$.
Proportional Approval Voting (PAV) is a Thiele rule with $w_i = 1/i$.
Approval Voting Chamberlin-Courant (CC) is a Thiele rule with $w = (1,0,0,\dots)$.

Our main results hold for a more general problem than \wthiele and, for clarity, we use a separate name:
the \genthiele problem generalizes \wthiele by allowing each voter to use a different Thiele sequence~\cite{SornatWX22_ijcai}:
\genthiele takes as input an approval election (with $n$ voters) and additionally a collection of $n$ Thiele sequences
$w \colon V \times \mathbb{N} \to [0,1]$,
where $(w^v_i)_{i \in \mathbb{N}}$ is a Thiele sequence of voter $v \in V$
(values of weights being in $[0,1]$ are due to normalization; note that we do not require $w^v_1$ to be $1$ for every $v \in V$).
For generality, we assume that all weights $(w_i^v)_{i \in [m], v \in V}$ are given as rational numbers, with numerators and denominators encoded in binary.
$\thiele(A,k,w)$ (and $\thiele(k)$ if $A$ and $w$ are clear from the context) is the maximum total score achieved by a winning committee of size $k$ on a given instance $(A,k,w)$, i.e.,
\[ \max_{W \subseteq C, |W|=k} \score_w(W), \quad\textrm{where } \score_w(W) \defeq \sum_{v \in V} \sum_{i=1}^{|A_v \cap W|} w_i^v.\]

An approval profile is \emph{Voter Interval} (VI)
if there exists an ordering of voters such that every candidate is approved by an interval of voters in this ordering~\cite{ElkindL15_ijcai15,ElkindLP2017trends_book,ElkindLP25_arxiv_preference_restrictions}.
Analogously, an approval profile is \emph{Candidate Interval} (CI)
if there exists an ordering of candidates such that every voter approves an interval of candidates in this ordering.
If a profile is VI (or CI), the respective ordering can be found in linear time~\cite{ElkindL15_ijcai15}.

Both VI and CI profiles can be interpreted as a family\footnote{Throughout the paper, a family is understood as an indexed collection.
In particular, different voters or candidates may correspond to identical approval ballots or supporters sets,
but they are still treated as distinct objects.}
of intervals over $V$ and $C$ respectively.
We define a few types of families of intervals which are used in this paper;
a family $U$ of intervals over $(1,2,\dots,n)$ is:
\begin{itemize}
  \item \emph{inclusion-free} if for every $I,I' \in U$ we have
    $I \not\subseteq I'$ and $I' \not\subseteq I$.
  \item \emph{dominancy-free} if for every $I,I' \in U$ such that $I \subseteq I'$ we have
    $\min(I) = \min(I')$ and $\max(I) = \max(I')$,
    i.e., $I = I'$ \cite{DongBWBE25_aamas25,LassotaSornat_aaai26}.
  \item \emph{endpoint-nested} if for every $I,I' \in U$ such that $I \subseteq I'$ we have
    $\min(I) = \min(I')$ or $\max(I) = \max(I')$.
\end{itemize}

Inclusion-freeness and dominancy-freeness can be defined on any family of sets (non necessarily intervals).
Every family of sets that is inclusion-free is also dominancy-free.
When considering a family of intervals,
by the definitions of these properties, additionally we have that
dominancy-freeness implies endpoint-nestedness.
Moreover, all these implications are strict.

A matrix $M \in \{-1, 0, 1\}^{p \times q}$ is \emph{totally unimodular} (TU) if every square submatrix of $M$ has determinant $-1$, $0$ or $1$.
A matrix $M \in \{0, 1\}^{p \times q}$ has the \emph{consecutive-ones property} (for columns) if its rows can be permuted so that, in every column, the entries equal to $1$ appear in a single consecutive block (possibly empty).
Note that a VI approval profile interpreted as a binary matrix from $\{0, 1\}^{V \times C}$ has consecutive-ones property.
It is a well known fact, that a binary matrix with consecutive-ones property is TU~\cite[p.~279]{Schrijver99_book}.
Moreover, every integer linear program with a TU constraint matrix has integral basic solutions which implies that a linear programming solver can find an integral optimal solution,
i.e., we can solve the ILP with TU constraint matrix in polynomial time.
For more on theory of linear programming we defer to~\cite{Schrijver99_book}.

In the remainder of this section we make a few observations.
First we give an equivalent definition of endpoint-nestedness.
\begin{lemma}\label{lem:en-eq-sorting-endpoints}
  A family $U$ of intervals over $(1,2,\dots,n)$ is endpoint-nested,
  if and only if,
  there exists an ordering $(I_1, I_2, \dots, I_{|U|})$ of intervals from $U$ such that:
  $\min(I_j) \le \min(I_{j+1})$ and
  $\max(I_j) \le \max(I_{j+1})$, for every $j \in [|U|-1]$.
\end{lemma}
\begin{proof}
  $(\implies)$ We define the ordering $\succeq$ of intervals from $U$ as follows:
  for $I,I' \in U$ we set $I \succ I'$ when either:
  1) $\min(I) < \min(I')$ and $\max(I) \le \min(I')$, or
  2) $\min(I) = \min(I')$ and $\max(I) < \min(I')$.
  Additionally, $I,I' \in U$ are in the same equivalency class with respect to $\succeq$ if $I=I'$.
  We note that the case $\min(I) < \min(I')$ and $\max(I) > \min(I')$ cannot appear because then $I' \subseteq I$ and endpoint-nestedness of $U$ requires that $\min(I) = \min(I')$ or $\max(I) = \max(I')$.

  Now, we show that the ordering $(I_1, I_2, \dots, I_{|U|})$ of intervals from $U$ defined by $\succeq$ (i.e., $I_j \succeq I_{j+1}$)
  has the desired property.
  We consider any $j \in [|U|]$.
  From conditions that defines $\succeq$ we have that
  $\min(I_j) < \min(I_{j+1})$ and $\max(I_j) \le \min(I_{j+1})$, or
  $\min(I_j) = \min(I_{j+1})$ and $\max(I_j) < \min(I_{j+1})$, or
  $\min(I_j) = \min(I_{j+1})$ and $\max(I_j) = \min(I_{j+1})$.
  In all cases it is true that
  $\min(I_j) \le \min(I_{j+1})$ and $\max(I_j) \le \min(I_{j+1})$,
  which finishes the proof of the first implication.

  $(\impliedby)$ Let $(I_1, I_2, \dots, I_{|U|})$ be an ordering of intervals from $U$ such that:
  $\min(I_j) \le \min(I_{j+1})$ and
  $\max(I_j) \le \max(I_{j+1})$, for every $j \in [|U|]$.
  Consider any $I_j,I_{j'} \in U$ where $j < j'$. We have $\min(I_j) \leq\min(I_{j'})$ and $\max(I_j) \leq\max(I_{j'})$. Thus, if $I_j \subseteq I_{j'}$, we must have $\min(I_j) = \min(I_{j'})$. Similarly, if $I_{j'} \subseteq I_{j}$, then $\max(I_j) = \max(I_{j'})$. Hence, $U$ is endpoint-nested.
\end{proof}

We introduce a new structured domain of approval profiles, based on the endpoint-nestedness property.

\begin{definition}
  An approval profile $A$ is \emph{Endpoint-Nested} (EN)
  if there exists an ordering of voters such that
  the family of sets of supporters $(V_c)_{c \in C}$ is endpoint-nested over the ordering.
\end{definition}

It is known that a VI profile that is dominancy-free is also CI~\cite{DongBWBE25_aamas25}.
This leads to a polynomial time algorithm for \genthiele on VI dominancy-free profiles using the algorithm for CI profiles~\cite{Peters18_aaai,PetersL20_spoc_jair}.
Below we show that EN profiles,
that form a strict superset of VI dominancy-free profiles,
are also CI, so \genthiele can be solved in polynomial time on EN profiles (\Cref{cor:thiele-en-poly}).
\begin{theorem}\label{thm:en-implies-ci}
  An Endpoint-Nested approval profile is Candidate Interval.
\end{theorem}
\begin{proof}
Let $A$ be an EN approval profile and let $(v_1,v_2,\dots,v_n)$ be the corresponding ordering of voters.
Since EN is VI and we fixed the ordering of voters, instead of $(v_1,v_2,\dots,v_n)$ and $V$ we will use simply voters' indices,
i.e., $(1,2,\dots,n)$ and $[n]$ to denote the order of voters and the set of voters, respectively.

For any arbitrary voter $i \in [n]$,
we have $c \in A_i$
if and only if
$\min(V_c) \le i \le \max(V_c)$.
Using \Cref{lem:en-eq-sorting-endpoints} we can also order candidates into sequence $(c_1,c_2,\dots,c_m)$ such that
$\min(V_{c_j}) \le \min(V_{c_{j+1}})$ and
$\max(V_{c_j}) \le \max(V_{c_{j+1}})$, for every $j \in [m-1]$.
It means that both left and right endpoint sequences are non-decreasing.
This implies that the condition $\min(V_{c_j}) \le i$ defines a prefix of indices $\{1, \dots, p\}$,
and the condition $\max(V_c) \ge i$ defines a suffix of indices $\{q, \dots, m\}$.
Thus, the candidates approved by voter $i$ correspond exactly to a contiguous block of indices $\{q, q+1, \dots, p\}$,
therefore $A$ is CI.
\end{proof}

\begin{corollary}\label{cor:thiele-en-poly}
  \genthiele on Endpoint-Nested profiles can be solved in polynomial time.
\end{corollary}

\section{Concavity of the Objective Function}\label{sec:concavity}

In this section we prove that, for a given VI approval profile $A$, $\thiele(k)$ is concave, i.e.,
the maximum achievable total score is a concave function of the committee size (\Cref{cor:thiele-concave}).
This follows from \Cref{thm:combining-solutions} which actually states a more general result and we believe it is of independent interest (therefore we stated it in terms of families of intervals).

One of the instantiations of \Cref{thm:combining-solutions} is a polynomial-time procedure for combining two solutions of sizes $s$ and $s'$ (not necessarily optimal) from a VI profile,
where $s+s'$ is even,
into a solution
(being a subset of union of these two solutions)
of size $(s+s')/2$
that achieves the objective value at least the average of the scores of these two solutions.

\Cref{thm:combining-solutions} is stated in general terms of non-decreasing, concave functions $f_i$'s.
In order to see that this applies to $\genthiele$
it is enough to define $f_v(j) = \sum_{i=1}^j w_i^v$ which, indeed, are non-decreasing and concave functions because $w_j^v \ge w_{j+1}^v$ for all $v \in V, j \in \mathbb{N}$.
Then we have $F(X) = \score_w(X)$ since $c_X(i) = |A_i \cap X|$. 
Note that a solution (winning committee) to \genthiele on VI is a family of intervals (corresponding to candidates selected to the committee).

Recall that we treat families of intervals as multisets,
and so, we interpret set operators such as $\cup$, $\cap$, and $\setminus$ as the corresponding multiset union/intersection/difference (i.e., with multiplicities).

\begin{theorem}\label{thm:combining-solutions}
  Let $f_1, \dots, f_n : \mathbb{Z}_{\ge 0} \to \mathbb{R}_{\ge 0}$ be non-decreasing, concave functions.
  For any family $X$ of intervals over $(1,2,\dots,n)$,
  let $F(X) \defeq \sum_{i=1}^n f_i(c_X(i))$ where $c_X(i) := |\{ I \in X : i \in I\}|$ is the \emph{coverage} at point $i \in [n]$.

Given two families of intervals $S$ and $S'$ over $(1,2,\dots,n)$,
such that $|S| < |S'|$,
there is a polynomial-time algorithm
that computes $P = R \cup R' \cup (S \cap S')$ of size
$|P| \in \left\{\left\lfloor \frac{|S|+|S'|}{2}\right\rfloor, \left\lceil \frac{|S|+|S'|}{2}\right\rceil\right\}$
such that
$R \subseteq S \setminus S', R' \subseteq S' \setminus S$ and
\[
  F(P) \ge (1-\theta) \cdot F(S) + \theta \cdot F(S'),
  \qquad \text{where } \theta \defeq \frac{|P|-|S|}{|S'|-|S|}.
\]
\end{theorem}

\begin{proof}
First we prove the theorem for the case that $S \cap S' = \emptyset$.
In the last part of the proof, we show how to handle the case $S \cap S' \ne \emptyset$.

\paragraph{The Case $S \cap S' = \emptyset$.}
Let $U \defeq S \cup S'$ be the union of the given families, which has a total size of $u \defeq |S|+|S'|$.

\paragraph{Step 1: The Canonical Endpoint-Nested Family $U^*$.}
Extract all $u$ left endpoints and $u$ right endpoints from the intervals in $U$,
and independently sort them into non-decreasing sequences:
\begin{align*}
L &= (l_1, l_2, \dots, l_{u}) &\textrm{such that} \quad &l_1 \le l_2 \le \dots \le l_{u},\\
R &= (r_1, r_2, \dots, r_{u}) &\textrm{such that} \quad &r_1 \le r_2 \le \dots \le r_{u}.
\end{align*}

We define a \emph{canonical family} of intervals $U^* \defeq \{ J_j = [l_j, r_j] \}_{j=1}^{u}$ by pairing the $j$-th smallest left and right endpoints. 
\Cref{claim:ustar-proper-int} verifies validity of definition of $U^*$ by showing that the intervals $\{ J_j = [l_j, r_j] \}_{j=1}^{u}$ are well defined.

\begin{claim}\label{claim:ustar-proper-int}
  For all $j \in [u]$ it holds that $l_j \le r_j$.
\end{claim}
\begin{proof}[Proof of \Cref{claim:ustar-proper-int}]
  For $i \in [n]$ we define $U_L(i)$ as the number of intervals in $U$
  that have left endpoint being strictly less than $i$, i.e.,
  $U_L(i) = \{ I \in U : \min(I) < i \}$.
  Analogously for right endpoints we define $U_R(i) = \{ I \in U : \max(I) < i \}$.
  For every $i \in [n]$ we have
  $U_R(i) \subseteq U_L(i)$ since for every $I \in U_R(i)$ we have
  $i > \max(I) \ge \min(I)$,
  hence also $I \in U_L(i)$.
  Now, for any $j \in [u]$, $l_j$ is the $j$-th smallest left endpoint in $U$.
  It implies that at most $j-1$ intervals in $U$ have left endpoint being strictly less than $l_j$, hence
  $j-1 \ge |U_L(l_j)| \ge |U_R(l_j)|$.
  It means that also at most $j-1$ intervals in $U$ have right endpoint being strictly less than $l_j$,
  therefore the $j$-th smallest right endpoint is at least $l_j$, i.e.,
  $r_j \geq l_j$.
\end{proof}

We note that the canonical family of intervals $U^*$ is endpoint-nested since the definition of $U^*$ matches the condition of \Cref{lem:en-eq-sorting-endpoints}.

\paragraph{Step 2: Alternating Partition of $U^*$.}
Since $U^*$ is endpoint-nested,
using the same arguments as in the proof of \Cref{thm:en-implies-ci} we obtain that
the intervals in $U^*$ covering $i$ correspond exactly to a contiguous block of indices $\{q, q+1, \dots, p\}$, for some $q,p \in [u]$.

We partition $U^*$ into two families $A^*$ and $B^*$ with sizes
$\lfloor u/2 \rfloor$ and $\lceil u/2 \rceil$ respectively,
according to index parity:
\begin{align*}
    A^* &\defeq \{J_2, J_4, \dots\} = \{J_j \in U^*: j \text{ is even}\},\\
    B^* &\defeq \{J_1, J_3, \dots\} = \{J_j \in U^*: j \text{ is odd}\}.
\end{align*}
We call such $(A^*,B^*)$ an \emph{alternating partition of $U^*$}.
Because any contiguous sequence of integers alternates between even and odd,
the coverage made by intervals from $A^*$ and $B^*$ at any point $i \in [n]$ must differ by at most $1$:
\begin{align}
  |c_{A^*}(i) - c_{B^*}(i)| \le 1.\label{eq:astar-bstar-balanced}
\end{align}

\paragraph{Step 3: The Alternating Partition Yields the Convexity Bound.}
Since $(S,S')$ is a partition of $U$ and $(A^*,B^*)$ is a partition of $U^*$, we have
$c_{S}(i) + c_{S'}(i) = c_{U}(i)$
and
$c_{A^*}(i) + c_{B^*}(i) = c_{U^*}(i)$.
Furthermore, because $U$ and $U^*$ share the same endpoint sets $L$ and $R$,
their total point-wise coverage is identical:
$c_{U^*}(i) = c_U(i)$.
Putting these equalities together we obtain
$c_{S}(i) + c_{S'}(i) = c_{U}(i) = c_{A^*}(i) + c_{B^*}(i)$,
for every $i \in [n]$.

For every $i \in [n]$, the function
$x \mapsto f_i(x) + f_i(c_{U}(i)-x)$
is concave (since $f_i$ is concave)
and symmetric around $c_{U}(i)/2$.
Therefore, it is maximized when the split of $c_{U}(i)$ into $(x,c_{U}(i)-x)$ is as balanced as possible\footnote{This inequality is in fact a special case of Karamata's inequality~\cite{karamata32,marshall2011inequalities}.
}, i.e.,
when the two arguments in $(x,c_{U}(i)-x)$ differ by at most $1$.
Since this property holds for $(A^*,B^*)$ (due to~\Cref{eq:astar-bstar-balanced}),
we get
\[
  f_i(c_{A^*}(i)) + f_i(c_{B^*}(i)) = \max_{x \in \{0,1,\dots,c_{U}(i)\}} \big(f_i(x) + f_i(c_{U}(i)-x)\big) \ge f_i(c_S(i)) + f_i(c_{S'}(i)), \quad \text{for all } i\in [n].
\]
Summing the above inequality over all $i \in [n]$,
we obtain 
\begin{align}
  F(A^*) + F(B^*) \ge F(S) + F(S').\label{eq:astar-bstar-optimal}
\end{align}
Let
$\theta_{A^*} \defeq \frac{|A^*|-|S|}{|S'|-|S|}$
and
$\theta_{B^*} \defeq \frac{|B^*|-|S|}{|S'|-|S|}$.
Since \(|A^*|+|B^*| = u = |S|+|S'|\), we have
\begin{align}
  \theta_{A^*} + \theta_{B^*} = 1.\label{eq:lambdas-sums-to-one}
\end{align}
If
$F(A^*) \ge (1-\theta_{A^*}) \cdot F(S) + \theta_{A^*} \cdot F(S')$,
then we set $P^* \defeq A^*$.
Otherwise,
$F(A^*) < (1-\theta_{A^*}) \cdot F(S) + \theta_{A^*} \cdot F(S')$
and we obtain
\begin{align*}
  F(B^*)
  &\stackrel{\eqref{eq:astar-bstar-optimal}}{\ge} F(S)+F(S')-F(A^*)\\
  &> F(S)+F(S')-\big((1-\theta_{A^*}) \cdot F(S)+\theta_{A^*} \cdot F(S')\big)\\
  &= \theta_{A^*} \cdot F(S)+(1-\theta_{A^*}) \cdot F(S')\\ &
  \stackrel{\eqref{eq:lambdas-sums-to-one}}{=} (1-\theta_{B^*}) \cdot F(S)+\theta_{B^*} \cdot F(S').
\end{align*}
Hence, in this case we set $P^* \defeq B^*$.
Therefore, there exists \(P^* \in \{A^*,B^*\}\) which satisfies the inequality from the theorem statement, i.e.,
\begin{equation}
  F(P^*) \ge (1-\theta_{P^*}) \cdot F(S) + \theta_{P^*} \cdot F(S'),
  \qquad \text{where } \theta_{P^*} \defeq \frac{|P^*|-|S|}{|S'|-|S|}.\label{eq:pstar-concavity-bound}
\end{equation}
Moreover,
$|P^*| \in \{|A^*|,|B^*|\} = \{\lfloor u/2 \rfloor, \lceil u/2 \rceil\}$ as desired.
This does not finish the proof of the theorem because $P^* \subseteq U^*$, i.e.,
$P^*$ is a subset of canonical family, but we need to construct a subset of original input intervals $U = S \cup S'$ (recall that we consider the case $S \cap S' = \emptyset$).
In the next step, we will show how to transform any subset of $U^*$ to a subset of $U$ without decreasing its coverage at any point $i \in [n]$.

\paragraph{Step 4: Transformation Back to $U$.}
In order to construct $P \subseteq U$ based on $P^* \subseteq U^*$ we rely on the following structural lemma,
which shows how to transform $U^*$,
the canonical form of $U$,
back into $U$ via swaps of right endpoints.
This lemma applies to any family of intervals but we state it in terms of $U$ and its canonical form $U^*$ for concreteness.

\begin{lemma}[Transformation via Right-Endpoint Swaps]\label{lem:trans-right-endpoint-swap}
Any family of intervals $U$ can be obtained from its canonical endpoint-nested form $U^*$ by a sequence of at most $O(|U|^2)$ \emph{right-endpoint swap operations},
where such an operation replaces two intervals
$I_1 = [l_a, r_a]$ and $I_2 = [l_b, r_b]$
satisfying $l_a \le l_b \le r_a \le r_b$
with the two intervals
$I_\cup = [l_a, r_b]$ and $I_\cap = [l_b, r_a]$,
that is, with their union and intersection.
\end{lemma}
\begin{proof}[Proof of \Cref{lem:trans-right-endpoint-swap}]
We show the reverse:
$U$ can be transformed into its canonical endpoint-nested form $U^*$ by a finite sequence of right-endpoint swap operations.
Consider the intervals in $U$ ordered by their left endpoints $l_1 \le \dots \le l_{|U|}$.
Let $r'_1, \dots, r'_{|U|}$ be their corresponding right endpoints.
If $U \neq U^*$\footnote{Note that every interval family $U$ has a unique canonical form $U^*$.},
the right endpoints are not perfectly sorted,
meaning there exists an adjacent inversion:
an index $i$ such that $r'_i > r'_{i+1}$.
Because $l_i \le l_{i+1} \le r'_{i+1} < r'_i$,
the interval $[l_{i+1},r'_{i+1}]$ is strictly contained in $[l_i,r'_i]$.
A right-endpoint swap operation replaces them with $[l_i, r'_{i+1}]$ and $[l_{i+1}, r'_i]$.
This strictly decreases the number of inversions in the sequence of right endpoints.
Since the number of inversions is at most $O(|U|^2)$,
this process must terminate,
which occurs precisely when the right endpoints become sorted,
i.e., when the resulting family is $U^*$.
Reversing this process provides the required sequence of right-endpoint swap operations from $U^*$ to $U$.
\end{proof}

Let $U^* = U_0 \to U_1 \to \dots \to U_T = U$ be the transformation sequence of right-endpoint swap operations from \Cref{lem:trans-right-endpoint-swap}.
We create a sequence of families $P_t \subseteq U_t$,
initialized with $P_0 = P^*$,
such that $|P_t| = |P^*|$ for every $t \in \{0,1,\dots,T\}$.
For each transition from $U_t$ to $U_{t+1}$ replacing $\{I_1, I_2\} \subseteq U_t$ with $\{I_\cup, I_\cap\} \subseteq U_{t+1}$ we proceed as follows:
\begin{itemize}
    \item \textbf{Case 1:} If both $I_1, I_2 \in P_t$,
    we set $P_{t+1} = (P_t \setminus \{I_1, I_2\}) \cup \{I_\cap, I_\cup\}$.
    The point-wise coverage contribution of both pairs is identical, i.e.,
    $c_{\{I_1,I_2\}}(i) = c_{\{I_\cup,I_\cap\}}(i)$ for every $i \in [n]$,
    so the point-wise coverage of $P_t$ and $P_{t+1}$ is the same,
    hence $F(P_{t+1}) = F(P_t)$.
    \item \textbf{Case 2:} If neither is in $P_t$, we set $P_{t+1} = P_t$,
    so $F(P_{t+1}) = F(P_t)$.
    \item \textbf{Case 3:} If exactly one is in $P_t$ (call it $I_a$), we set $P_{t+1} = (P_t \setminus \{I_a\}) \cup \{I_\cup\}$.
    Because $I_\cup$ contains $I_a$, the coverage at every point $i \in [n]$ does not decrease, i.e.,
    $c_{I_\cup}(i) \ge c_{I_a}(i)$,
    so $F(P_{t+1}) \ge F(P_t)$.
\end{itemize}
Since $|P_T| = |P^*|$, the constructed sequence of families $P_t \subseteq U_t$ satisfies:
\begin{align}
  F(P_T) &\ge \dots \ge F(P_1) \ge F(P_0) = F(P^*)\nonumber\\
         &\stackrel{\eqref{eq:pstar-concavity-bound}}{\ge} (1-\theta_{P_T}) \cdot F(S) + \theta_{P_T} \cdot F(S'),\qquad
  \text{where } \theta_{P_T} \defeq \frac{|P_T|-|S|}{|S'|-|S|}.\label{eq:p-sat-concavity}
\end{align}

Family $P \defeq P_T$ fulfills all constraints of the theorem statement:
$P$ satisfies the concavity inequality due to \Cref{eq:p-sat-concavity};
$|P| \in \{|A^*|,|B^*|\} = \{\lfloor u/2 \rfloor, \lceil u/2 \rceil\}$;
and $P = R \cup R' \cup (S \cap S')$ with
$R \defeq P \cap S \subseteq S \setminus S'$ and
$R' \defeq P \cap S' \subseteq S' \setminus S$.
This concludes the proof for the case $S \cap S' = \emptyset$.

\paragraph{The Case $S \cap S' \neq \emptyset$.}
Finally, consider the general case when $S \cap S' \neq \emptyset$.
Let $D \defeq S \cap S'$ and for every $i \in [n]$ we define
\[
  g_i(x) \defeq f_i(x+c_D(i)).
\]
Since each $f_i$ is non-decreasing and concave, so is each $g_i$.
Let $S_0 \defeq S \setminus D$ and $S'_0 \defeq S' \setminus D$.
Then $S_0 \cap S'_0 = \emptyset$ and $|S_0| < |S'_0|$.
Applying the disjoint-case construction to $S_0$ and $S'_0$ with functions $g_1,\dots,g_n$,
we obtain a family $P_0 \subseteq S_0 \cup S'_0$ such that
\begin{align}
  |P_0| \in \left\{ \left\lfloor \frac{|S_0|+|S'_0|}{2} \right\rfloor, \left\lceil \frac{|S_0|+|S'_0|}{2} \right\rceil \right\}\label{eq:p0-size}
\end{align}
and, for $G(X) \defeq \sum_{i=1}^n g_i(c_X(i))$ for a family $X$ of intervals over $(1,2,\dots,n)$, we have
\begin{align}
  G(P_0) \ge (1-\theta_{P_0}) \cdot G(S_0) + \theta_{P_0} \cdot G(S'_0),\qquad
  \text{where } \theta_{P_0} \defeq \frac{|P_0|-|S_0|}{|S'_0|-|S_0|}.\label{eq:g-on-p0-satisfy-concavity}
\end{align}
Now let $P \defeq P_0 \cup D = P_0 \cup (S \cap S')$,
so $P$ has the required form and size since
\begin{align*}
  |P| = |P_0|+|D| &\stackrel{\eqref{eq:p0-size}}{\in} \left\{\left\lfloor \frac{|S_0|+|S'_0|+2|D|}{2} \right\rfloor, \left\lceil \frac{|S_0|+|S'_0|+2|D|}{2} \right\rceil\right\}\\
  &\hspace{42pt}= \left\{ \left\lfloor \frac{|S|+|S'|}{2} \right\rfloor, \left\lceil \frac{|S|+|S'|}{2} \right\rceil \right\}.
\end{align*}
By the definition of $g_i(\cdot)$ and $c_X(i)$ we have
\begin{align}
  F(P) = F(P_0 \cup D)
      &= \sum_{i=1}^n f_i(c_{P_0 \cup D}(i))
      = \sum_{i=1}^n f_i(c_{P_0}(i) + c_{D}(i))
      = \sum_{i=1}^n g_i(c_{P_0}(i)) = G(P_0)\label{eq:f-on-p-eq-g-on-p0}
\end{align}
Using the analogous chain of equalities for $S = S_0 \cup D$ and $S' = S'_0 \cup D$, we obtain
\begin{align}
  F(S) = G(S_0) \qquad\text{and}\qquad F(S') = G(S'_0)\label{eq:f-on-s-eq-g-on-s0}
\end{align}
This implies
\begin{align}
  F(P) \stackrel{\eqref{eq:f-on-p-eq-g-on-p0}}{=}
         G(P_0)
       \stackrel{\eqref{eq:g-on-p0-satisfy-concavity}}{\ge}
         (1-\theta_{P_0}) \cdot G(S_0) + \theta_{P_0} \cdot G(S'_0)
       \stackrel{\eqref{eq:f-on-s-eq-g-on-s0}}{=}
         (1-\theta_{P_0}) \cdot F(S) + \theta_{P_0} \cdot F(S').\label{eq:f-on-p-sat-concavity-lambda0}
\end{align}
Moreover,
\begin{align}
  \theta_{P_0}
  \stackrel{\eqref{eq:g-on-p0-satisfy-concavity}}{=}
  \frac{|P_0|-|S_0|}{|S'_0|-|S_0|}
  = \frac{(|P_0|+|D|)-(|S_0|+|D|)}{(|S'_0|+|D|)-(|S_0|+|D|)}
  = \frac{|P|-|S|}{|S'|-|S|}.\label{eq:lambda0-eq-lambda}
\end{align}
Hence
\[
  F(P) \stackrel{\eqref{eq:f-on-p-sat-concavity-lambda0},\eqref{eq:lambda0-eq-lambda}}{\ge}
  (1-\theta) \cdot F(S) + \theta \cdot F(S'),\qquad
  \text{where } \theta \defeq \frac{|P|-|S|}{|S'|-|S|}.
\]
This completes the proof.
\end{proof}

For two families of intervals $S$ and $S'$ given,
by applying \Cref{thm:combining-solutions} recursively,
in polynomial time we can construct a family of solutions of all intermediate sizes $|S|, |S|+1, \dots, |S'|$ that satisfy the concavity property from \Cref{thm:combining-solutions}.
In order to construct a solution of a given size $p$, $|S| < p < |S'|$ it is enough to apply \Cref{thm:combining-solutions} iteratively at most $O(\log(|S'|-|S|))$ times (in the binary-search manner).
We state it as a separate proposition below and prove it for completeness.

\begin{proposition}\label{prop:combining-solutions-any-size}
Let $f_1, \dots, f_n, c_X, F, S, S'$ be as in \Cref{thm:combining-solutions}.
Then, for any $k \in [m]$
such that $|S| < k < |S'|$,
there is a polynomial-time algorithm
that computes $P = R \cup R' \cup (S \cap S')$ of size
$k$ such that
$R \subseteq S \setminus S', R' \subseteq S' \setminus S$ and
\[
  F(P) \ge (1-\theta) \cdot F(S) + \theta \cdot F(S'),
  \qquad \text{where } \theta \defeq \frac{|P|-|S|}{|S'|-|S|}.
\]
\end{proposition}
\begin{proof}
We define a recursive procedure on pairs $(X,Y)$ of families of intervals over $(1,2,\dots,n)$ satisfying
$S \cap S' \subseteq X \cap Y \subseteq S \cup S'$
and $|X| \le k \le |Y|$.
The procedure maintains the following invariant
(the inequality from the proposition statement)
for both families $X$ and $Y$:
\begin{equation}
  F(Z) \ge (1-\theta_Z) \cdot F(S) + \theta_Z \cdot F(S'),
  \qquad \text{where } Z \in \{X,Y\} \:\: \text{and } \theta_Z \defeq \frac{|Z|-|S|}{|S'|-|S|}.\label{eq:invariant-xy}
\end{equation}

Initially, we call the procedure on $(S,S')$
and the invariant holds with equality. 
We apply \Cref{thm:combining-solutions} to $X$ and $Y$.
It returns a family
$M = X' \cup Y' \cup (X \cap Y)$
such that
$X' \subseteq X \setminus Y, Y' \subseteq Y \setminus X$, and
$|M| \in \left\{\left\lfloor \frac{|X|+|Y|}{2} \right\rfloor, \left\lceil \frac{|X|+|Y|}{2}\right\rceil\right\}$.
Moreover,
\begin{equation}
  F(M) \ge (1-\eta) \cdot F(X) + \eta \cdot F(Y),
  \qquad \text{where } \eta \defeq \frac{|M|-|X|}{|Y|-|X|}.\label{eq:m-above-xy-interpolation}
\end{equation}
Since $S \cap S' \subseteq X \cap Y$,
we also have $S \cap S' \subseteq M$,
and clearly $M \subseteq X \cup Y \subseteq S \cup S'$.
Substituting the invariant bounds for $F(X)$ and $F(Y)$ in \Cref{eq:m-above-xy-interpolation} yields
\begin{equation}
  F(M)
  \stackrel{\eqref{eq:invariant-xy},\eqref{eq:m-above-xy-interpolation}}{\ge}
  (1-\eta) \cdot \bigl( (1-\theta_X) \cdot F(S) + \theta_X \cdot F(S') \bigr)
  +
  \eta \cdot \bigl( (1-\theta_Y) \cdot F(S) + \theta_Y \cdot F(S')\bigr).\label{eq:m-above-xy-in-relation-to-s-sprim}
\end{equation}
A direct calculation shows that
$(1-\eta) \cdot \theta_X + \eta \cdot \theta_Y = \frac{|M|-|S|}{|S'|-|S|}$,
and therefore
\[
  F(M)
  \stackrel{\eqref{eq:m-above-xy-in-relation-to-s-sprim}}{\ge}
  (1-\theta_M)\cdot F(S)+\theta_M\cdot F(S'),
  \qquad
  \theta_M \defeq \frac{|M|-|S|}{|S'|-|S|}.
\]
Thus the invariant \eqref{eq:invariant-xy} is preserved if we replace one of the endpoints by $M$:
if $k < |M|$, we recurse on $(X,M)$;
if $k > |M|$, we recurse on $(M,Y)$;
otherwise, if $|M| = k$, we simply return $P \defeq M$.

In each recursive step, the interval of feasible sizes shrinks from length $|Y|-|X|$ to at most
$\left\lceil \frac{|Y|-|X|}{2} \right\rceil$.
Hence after at most $O(\log(|S'|-|S|))$ recursive calls we reach a family $M^*$ of size exactly $k$.
By construction, $S \cap S' \subseteq M^* \subseteq S \cup S'$.
Thus, if we define $P \defeq M^*$ and
$R \defeq M^* \cap (S \setminus S'), R' \defeq M^* \cap (S' \setminus S)$,
then $P = R \cup R' \cup (S \cap S')$,
with $R \subseteq S \setminus S'$ and $R' \subseteq S' \setminus S$ as required.
Finally, the invariant \eqref{eq:invariant-xy} for the family $M^*$ is equivalent to the inequality from the theorem statement on $P \defeq M^*$.
This finishes the proof.
\end{proof}

By applying \Cref{thm:combining-solutions} on two optimal solutions $S$ and $S'$ of sizes $k-1$ and $k+1$ respectively,
we obtain concavity of the objective $\thiele(k)$ as a function of the committee size $k$.
\begin{corollary}\label{cor:thiele-concave}
  For a given approval profile $A$ and a collection of Thiele sequences $w$,
  the maximum total score achieved by a winning committee of size $k$ under the \genthiele rule is a concave function of $k$, i.e.,
  for every $k \geq 1$ we have
  \[ \thiele(k) \geq \frac{\thiele(k+1)+\thiele(k-1)}{2}. \]
\end{corollary}

\section{Polynomial-Time Algorithm}\label{sec:poly-time-alg}

In this section we prove our main result which implies a polynomial-time solvability of any Thiele voting rule, in particular PAV, on VI.

\begin{theorem}\label{thm:poly-thiele-vi}
  There is a polynomial-time algorithm for \genthiele on Voter Interval profiles.
\end{theorem}
\begin{proof}
We describe how to compute an optimal solution to \genthiele instance $(A,k,w)$, where $A$ is VI, in polynomial time.
Excluding trivial cases, we assume that $k \in \{1,2,\dots,m-1\}$.

\paragraph{Standard ILP.}
We use a standard Integer Linear Programming  (ILP) formulation of the problem~\cite{Peters18_aaai,PetersL20_spoc_jair,SornatWX22_ijcai,LassotaSornat_aaai26}:
\begin{align}
    \textrm{maximize}\quad  \sum_{v \in V} \sum_{i \in [k]} w^v_i   &\cdot x_{v,i}                   \label{ilp:objective}\\
    \textrm{subject to}\quad\quad\quad\hspace{0pt} \sum_{c \in C}                  y_c     &= k  \label{ilp:cardinality} \\
                        \sum_{i \in [k]}                x_{v,i} &\leq \sum_{c \in A_v} y_c       \quad\quad\;\forall v \in V \label{ilp:counting}\\
                        x_{v,i} &\in \{0,1\} \quad\quad\quad\forall v \in V, i \in [k] \nonumber\\
                        y_c     &\in \{0,1\} \quad\quad\quad\hspace{1pt}\forall c \in C, \nonumber
\end{align}
where selection variable $y_c = 1$ means selecting candidate $c$ to the committee and counting variable $x_{v,i} = 1$ means that voter $v$ approves at least $i$ selected committee members.
Cardinality constraint \eqref{ilp:cardinality} forces selecting exactly $k$ candidates.
The constraint \eqref{ilp:counting} for $v \in V$ transfers the number of selected approved candidates $\sum_{c \in A_v} y_c$ into counting variables $x_{v,i}$ which are used in the objective function \eqref{ilp:objective}.

\paragraph{Lagrangian Relaxation.}
We modify the ILP by removing the cardinality constraint $\sum_{c \in C} y_c = k$ (\Cref{ilp:cardinality}) and adding penalty $\lambda \cdot (k - \sum_{c \in C} y_c)$ to the objective function, for some value of Lagrange multiplier $\lambda > 0$.
This operation (called Lagrangian relaxation of a constraint) moves a constraint into the objective by penalizing violating the constraint,
where the penalty has linear dependence on the original constraint.
Additionally, we remove the constant term $\lambda \cdot k$ from the objective.
We call the resulting integer linear program as LR and we use notation $\LR{\lambda}$ when we mean LR with a specific value of Lagrange multiplier:
\begin{align}
    \textrm{maximize}\quad  \sum_{v \in V} \sum_{i \in [k]} w^v_i &\cdot x_{v,i} - \lambda \cdot \sum_{c \in C} y_c \label{new-ilp:objective}\\
    \textrm{subject to}\quad\hspace{13pt}\sum_{i \in [k]}                x_{v,i} &\leq \sum_{c \in A_v} y_c       \quad\quad\quad\quad\;\forall v \in V \nonumber\\
                        x_{v,i} &\in \{0,1\} \quad\quad\quad\quad\quad\forall v \in V, i \in [k] \nonumber\\
                        y_c     &\in \{0,1\} \quad\quad\quad\quad\quad\hspace{1pt}\forall c \in C \nonumber
\end{align}

The LR removes a problematic cardinality constraint: on VI profiles the remaining constraint matrix is TU (as shown in~\cite{Peters18_aaai}),
and thus the LR has an integral optimum and can be found in polynomial time.

\paragraph{Lagrange Multipliers.}
\Cref{cor:thiele-concave} says that the optimal score $\thiele(t)$ achieved by a committee of size $t$ is a concave function on VI profiles.
Therefore, every $\lambda \in [\thiele(k+1)-\thiele(k), \thiele(k)-\thiele(k-1)] \subseteq [0,n]$ is a slope of an affine function
$i \mapsto \thiele(k)+\lambda \cdot (i-k)$
that passes through the point $(k,F(k))$ and upper-bounds $\thiele(i)$ for all $i \in \{0,1,\dots,m\}$, i.e.,
\[
  \thiele(k) + \lambda \cdot (i-k) \ge \thiele(i),
\]
This is equivalent to
\begin{align}
  \thiele(k) - \lambda \cdot k \ge \thiele(i) - \lambda \cdot i \nonumber
\end{align}
which implies that the maximum value of $\LR{\lambda}$ objective \eqref{new-ilp:objective},
for every $\lambda \in [\thiele(k+1)-\thiele(k), \thiele(k)-\thiele(k-1)] \subseteq [0,n]$,
is achieved for a solution of size $k$.
We will search for a fixed Lagrangian multiplier from this interval, i.e., $\lambda_k \defeq \thiele(k) - \thiele(k-1)$.

\paragraph{Binary Search for $\bm{\lambda_k}$.}
In order to find $\lambda_k = \thiele(k) - \thiele(k-1) \in [0,n]$ we use binary search over $\lambda \in [0,n]$.
This leads us either to an optimal solution of size $k$ or to two solutions $S,S'$ such that $|S| < k < |S'|$.
This type of Lagrange-multiplier binary search and producing two bracketing solutions is standard, for example,
Lagrangian relaxation of the cardinality constraint in the \textsc{$k$-Median} problem is equivalent to the \textsc{Uncapacitated Facility Location} problem with uniform opening cost.
In this problem one searches for multipliers that bracket $k$ open facilities before combining the resulting solutions (see, e.g.,~\cite{JainV01_lagrangian,CohenAddadLSS25_kmedian}).
For completeness, we provide the full binary-search argument below.

For a tested value $\lambda$,
we solve $\LR{\lambda}$ (in polynomial time) and obtain an optimal committee $W_\lambda$.
If $|W_\lambda|=k$, we can stop and output $W_\lambda$.
Otherwise, we do \emph{not} interpret $|W_\lambda| \neq k$ as meaning that $\lambda$ is already incorrect---it may happen that $\lambda = \lambda_k$ and $\LR{\lambda_k}$ has optimal committees of different sizes,
and the solver returns one of them.
Instead, we maintain a bracketing interval $[\lambda^-, \lambda^+]$ with the property that
$\max \{T(\lambda^-)\} \ge k$ and
$\min \{ T(\lambda^+) \} \le k$,
where $T(\lambda)$ denotes the set of sizes of optimal committees for $\LR{\lambda}$.
Initially, this holds for $\lambda^- = 0$ and $\lambda^+ = n$.
At each step we test $\lambda = \frac{\lambda^{-} + \lambda^{+}}{2}$;
if $|W_\lambda| > k$ we set $\lambda^{-} \leftarrow \lambda$, and if $|W_\lambda| < k$ we set $\lambda^{+} \leftarrow \lambda$.

This update rule remains valid since
as $\lambda$ increases, larger committees are penalized more by the term $-\lambda \cdot \sum_{c \in C} y_c$,
and thus both
$\min\{T(\lambda)\}$ and $\max\{T(\lambda)\}$ are monotone non-increasing functions of~$\lambda$.
Consequently, observing an optimal solution of size strictly greater than $k$ certifies that the considered multiplier $\lambda$ is (weakly) too small,
while observing an optimal solution of size strictly smaller than $k$ certifies that the considered multiplier $\lambda$ is (weakly) too large,
which justifies maintaining a bracketing interval in the binary search.

We continue the binary search until the length of the search interval is below additive precision $\varepsilon$ (as defined below), i.e., $\lambda^{+}-\lambda^{-} \le \varepsilon$.
(Recall that if at any point we obtain an optimal committee $W_\lambda$ with $|W_\lambda|=k$, we stop and output it.)

To bound the number of iterations,
recall that all weights $(w_i^v)_{i \in [m], v \in V}$ are given as rational numbers, with numerators and denominators encoded in binary.
Let $D$ be a common denominator of all input weights
(e.g., the least common multiple of their denominators).
Then every committee total score is a multiple of $1/D$.
Every breakpoint slope $\lambda$ at which the set of optimal sizes $T(\lambda)$ changes is of the form
$\thiele(t)-\thiele(t-1)$, for $t \in [m]$,
therefore it is also a multiple of $1/D$.
Hence any two distinct breakpoints differ by at least $1/D$ and
it suffices to run the binary search to additive precision
$\frac{1}{2D}$.
In particular, letting $B$ denote the total bit-length of the denominators in the input,
we have $D \le 2^{B}$ and may take
$\varepsilon \defeq 1/2^{B+1}$,
which yields $O(\log(n/\varepsilon))=O(\log n + B)$ iterations.
Since the input size is $\Omega(nm + B)$,
the number of binary-search iterations is at most linear in the input size.

Now, we describe the next steps of the algorithm
(assuming that during the binary-search procedure a committee of size $k$ has not been found).
Let $\bar{\lambda} \defeq (\lambda^- + \lambda^+)/2$ and solve $\LR{\bar{\lambda}}$,
obtaining a committee $W_{\bar{\lambda}}$.
If $|W_{\bar{\lambda}}|=k$, we stop and output $W_{\bar{\lambda}}$.

Otherwise, since $\lambda^+-\lambda^- \le \varepsilon$ and
$\varepsilon$ is chosen smaller than the minimum distance between two distinct breakpoints,
the interval $[\lambda^-,\lambda^+]$ contains at most one breakpoint.
Moreover, by construction of the bracketing interval,
the unique breakpoint in $[\lambda^-,\lambda^+]$ must be $\lambda_k$,
and there is no open interval of multipliers for which an LR-optimal committee has size exactly $k$.
Equivalently, $k$ lies on a linear segment of $\thiele(\cdot)$ (i.e., $|T(\lambda_k)|>1$).

In this case we perturb the multiplier and solve $\LR{\bar{\lambda}-\varepsilon}$ and $\LR{\bar{\lambda}+\varepsilon}$,
obtaining two LR-optimal committees $S,S'$ such that $|S|<k<|S'|$.
Intuitively, this perturbation selects the two endpoints of the size-interval $T(\lambda_k)$:
decreasing $\lambda$ slightly weakens the size penalty and yields a larger optimal committee, 
while increasing $\lambda$ slightly strengthens the penalty and yields a smaller optimal committee.
By concavity, this implies that $\thiele(t)$ is linear on the whole interval of sizes between $|S|$ and $|S'|$, and in particular
\begin{align}
  \thiele(a)-\lambda_k \cdot a = \thiele(k)-\lambda_k \cdot k,\qquad \text{for every } a \in \{|S|, |S|+1, \dots, |S'|\}.
  \label{eq:lambda-k-many-sizes}
\end{align}

Now, using \Cref{prop:combining-solutions-any-size} to committees $S$ and $S'$,
we can construct a committee $W$ of size exactly $k$
(after at most $\log(|S'|-|S|) \le \log(m)$ calls to \Cref{thm:combining-solutions})
which for $\theta \defeq \frac{k-|S|}{|S'|-|S|}$ satisfies the following:
\begin{align*}
  F(W) &\ge (1-\theta) \cdot F(S) + \theta \cdot F(S')\\
  &= (1-\theta) \cdot (\thiele(|S|)-\lambda_k \cdot |S| ) + \theta \cdot (\thiele(|S'|)-\lambda_k \cdot |S'| )\\
  &\stackrel{\eqref{eq:lambda-k-many-sizes}}{=}
  (1-\theta) \cdot (\thiele(k)-\lambda_k \cdot k ) + \theta \cdot (\thiele(k)-\lambda_k \cdot k ) = \thiele(k)-\lambda_k \cdot k.
\end{align*}
This inequality means that $W$ achieves a maximal objective value among all size-$k$ committees, so it is an optimal solution of the \genthiele instance.
This finishes the proof.
\end{proof}

\section{AI Assistance}\label{sec:ai-assistance}

In this section, we describe how Gemini (Deep Think), an AI-assistance tool, helped in achieving the results presented in this paper.
Analogous descriptions appeared recently~\cite{woodruff2026acceleratingscientificresearchgemini,Knuth2026ClaudesCycles},
where results in theoretical computer science were achieved with AI-assistance. We had two separate interactions with Gemini~\cite{Gemini-tech-report2025}.

\paragraph{Initial Interactions: Arriving at Lagrangian Relaxation.}
As a first step, we use Gemini Pro to generate the main ideas for the algorithm.
In particular, we simply prompt the model to solve the problem, using the full problem description.
Originally, the model proposed an incorrect polynomial-time algorithm based on Maximum Flow.
Fortunately, once we point out to the model that this algorithm is incorrect, 
the model immediately suggests the Lagrangian relaxation approach,
which is essentially the same as the one we eventually used.
However, once we ask the model to provide the proof of correctness, it provides an incorrect proof.
In particular, the concavity argument is incorrect.
After several attempts to prompt the model to correct this proof, the model still fails.

The sequence of interactions we have with the model---with full prompts but with responses abbreviated---is given in Appendix~\ref{app:ai-lagrangian}.

\paragraph{Final Step: Proof of Concavity of the Objective.}
As concavity of the objective function was a crucial missing part of the Lagrangian relaxation approach, in a second step, we reformulate the concavity property in a self-contained manner.
Using this reformulation, we prompt Gemini Deep Think.
In the very first response, the model gives us the correct proof of the statement, which is the basis for our main concavity theorem (\Cref{thm:combining-solutions}).
We remark our eventual theorem is slightly more generalized than the version we asked the model.
Specifically, we only ask the model to combine two solutions whose size differ by exactly two and did not ask for a polynomial time algorithm.
Moreover, the concavity formulation assumes that the two solutions are disjoint,
since their common elements might be included in the combined solution.
Formally, the solution created by the model might use the same interval twice,
which is not allowed in the final solution.
The model did not specifically discuss this issue.
Nevertheless, our final proof is essentially the same as produced by the model.

The exact user prompt and response for this second step are given in Appendix~\ref{app:ai-concavity}.

\section{Conclusion}

In this work we gave a polynomial-time algorithm for computing an optimal size-$k$ committee for Thiele voting rules on the Voter Interval domain,
thereby resolving a 10-year-old open problem \cite{ElkindL15_ijcai15,ElkindLP2017trends_book,Peters18_aaai,YangW18-aamas18,PetersL20_spoc_jair,GodziszewskiB0F21_vci_aaai21,lackner23abc_book,YangW23-jaamas,ElkindLP25_arxiv_preference_restrictions,LassotaSornat_aaai26}.
We also extended the result to the Generalized Thiele rule.
A key ingredient is our concavity theorem for families of intervals, which we view as a structural result of independent interest.

An interesting direction for future work is to obtain a purely combinatorial algorithm for Thiele rules on the Voter Interval domain (and, analogously, on the Candidate Interval domain~\cite{Peters18_aaai,PetersL20_spoc_jair}), avoiding linear-programming solvers.
Achieving this may require further structural insights beyond concavity,
which could deepen our understanding of why these domains admit efficient exact algorithms and may be of independent interest.

Finally, our results illustrate how human--AI collaboration can help identify useful structural conjectures and guide proof development in theoretical computer science.
We hope this experience will encourage further systematic study of such workflows on other open problems in theoretical computer science.

\subsection*{Acknowledgements}
Krzysztof Sornat was supported by the European Research Council (ERC) under the European Union’s Horizon 2020 research and innovation programme (grant agreement No 101002854).

\bibliographystyle{alpha}
\bibliography{bib}

\appendix

\iftrue
\section{Human--AI Interaction Log I: Proof of Concavity}\label{app:ai-concavity}

\setcounter{secnumdepth}{0}

Below, we include the exact user prompt and response from the model that provides the proof of concavity (\Cref{thm:combining-solutions}).
We stress that this is a single call to Gemini Deep Think without further interaction.
The response is a correct proof of the required statement.

\begin{prompt}
Consider the following objective. We have non-decreasing functions $f_1, \dots, f_n$ from non-negative integers to non-negative reals such that they are concave. For given a multiset of intervals $S = \{I_1, \dots, I_k\}$, its *total value* is computed as follows: For every $i \in [n]$, let $c_i$ denote the number of intervals containing $i$. Then, let the total value be $F(S) = f_1(c_1) + \cdots + f_k(c_k)$. Prove the following statement: Given two multiset of intervals $S, S'$ of size $k - 1, k + 1$ respectively (where $k$ is a positive integer). Then, we can always find a sub-multiset $S'' \subseteq S \cup S'$ of size $k$ such that $F(S'') \geq (F(S) + F(S')) / 2$.
\end{prompt}
\begin{answer}
To prove this statement, we can use a constructive combinatorial approach. The problem asks us to find a sub-multiset $S'' \subseteq S \cup S'$ of size $k$ whose total value is at least the average of the values of $S$ and $S'$.

Let $U = S \cup S'$ be the multiset union of the two given collections of intervals. The size of $U$ is $(k - 1) + (k + 1) = 2k$.

\paragraph{Step 1: The Canonical Uncrossed Multiset $U^*$.}

First, we extract all $2k$ left endpoints and $2k$ right endpoints from the intervals in $U$. We sort both multisets independently in non-decreasing order:

Left endpoints $L$: $l_1 \le l_2 \le \dots \le l_{2k}$
Right endpoints $R$: $r_1 \le r_2 \le \dots \le r_{2k}$
We define a canonical multiset of intervals $U^* = \{ J_m = [l_m, r_m] \}_{m=1}^{2k}$ by pairing the $m$-th smallest left endpoint with the $m$-th smallest right endpoint.

Validity: For any $m$, there are exactly $m-1$ left endpoints strictly less than $l_m$. Thus, at most $m-1$ intervals in $U$ can finish strictly before $l_m$. This implies at least $2k - (m-1)$ right endpoints in $U$ must be $\ge l_m$. Consequently, the $m$-th smallest right endpoint $r_m$ must be $\ge l_m$. Therefore, all intervals in $U^*$ are strictly valid.
Identical Total Coverage: Because $U$ and $U^*$ share the exact same sets of endpoints $L$ and $R$, their total pointwise coverage is identical. For any point $x \in [n]$, $c_{U^*}(x) = c_U(x) = c_S(x) + c_{S'}(x)$.

\paragraph{Step 2: Alternating Partition of $U^*$.}

Consider the ordered sequence of intervals $J_1, J_2, \dots, J_{2k}$ in $U^*$. For any arbitrary point $x$, an interval $J_m$ covers $x$ if and only if $l_m \le x \le r_m$. Because $l_m$ is non-decreasing, the condition $l_m \le x$ is satisfied by a prefix of indices $\{1, 2, \dots, p\}$. Because $r_m$ is non-decreasing, $r_m \ge x$ is satisfied by a suffix of indices $\{q, q+1, \dots, 2k\}$. Therefore, the subset of intervals in $U^*$ covering $x$ corresponds exactly to a contiguous subsegment of indices $\{q, q+1, \dots, p\}$.

We partition $U^*$ into two sub-multisets of size $k$ by alternating the sorted indices:
$A^* = \{J_1, J_3, J_5, \dots, J_{2k-1}\}$,
$B^* = \{J_2, J_4, J_6, \dots, J_{2k}\}$.
Because any contiguous block of integers is split as evenly as mathematically possible into evens and odds, the number of covering intervals at point $x$ distributed to $A^*$ and $B^*$ will differ by at most 1:
\[
|c_{A^*}(x) - c_{B^*}(x)| \le 1
\]

\paragraph{Step 3: Majorization and Lower Bound.}

For every point $x$, the integers $c_{A^*}(x)$ and $c_{B^*}(x)$ sum to $c_S(x) + c_{S'}(x)$ and differ by at most 1. This means they form the most balanced possible pair of integers for that sum.

By the Hardy-Littlewood-Pólya majorization inequality, a more balanced distribution maximizes the sum of any concave functions. Since each $f_x$ is concave:
\[
f_x(c_{A^*}(x)) + f_x(c_{B^*}(x)) \ge f_x(c_S(x)) + f_x(c_{S'}(x))
\]

Summing this across all points $x \in [n]$ yields:
\[
F(A^*) + F(B^*) \ge F(S) + F(S')
\]

Since $A^*$ and $B^*$ both have a size of exactly $k$, the Pigeonhole Principle dictates that at least one of them (assume $A^*$ without loss of generality) satisfies:
\[
F(A^*) \ge \frac{F(S) + F(S')}{2}
\]

\paragraph{Step 4: Transformation Back to $U$.}

We have found a size-$k$ subset $A^*$ from $M^*$ that achieves the required threshold, but we need a subset $S''$ drawn from our original multiset $U$.

Any original configuration $U$ can be reached from the canonical $U^*$ via a finite sequence of "crossing" operations. A crossing operation takes two crossing intervals $I_1 = [l_a, r_a]$ and $I_2 = [l_b, r_b]$ (where $l_a \le l_b \le r_a \le r_b$) and replaces them with their union $I_{\cup} = [l_a, r_b]$ and intersection $I_{\cap} = [l_b, r_a]$. (This is equivalent to reversing the standard "uncrossing" algorithm, meaning it takes adjacent inversions of the right-endpoint permutation and sorts them back into the configuration found in $U$).

Let $U^* = M_0 \to M_1 \to \dots \to M_N = U$ be this sequence. We can maintain a sub-multiset $Z_i \subseteq M_i$ of size $k$ at every step starting with $Z_0 = A^*$:

If $Z_i$ contains both $I_1, I_2$, we place both $I_{\cup}, I_{\cap}$ into $Z_{i+1}$. The sum of their point-wise coverages is perfectly identical since $1_{I_1} + 1_{I_2} = 1_{I_{\cup}} + 1_{I_{\cap}}$.
If $Z_i$ contains neither, we put neither into $Z_{i+1}$.
If $Z_i$ contains exactly one of $\{I_1, I_2\}$, we strictly place the union $I_{\cup}$ into $Z_{i+1}$. Because the union strictly covers the spans of both $I_1$ and $I_2$, the coverage at any point either strictly increases or remains the same.
By induction, we arrive at a size-$k$ sub-multiset $Z_N \subseteq U$. Because coverage never decreases at any step and the given functions $f_x$ are non-decreasing, evaluations only increase or stay constant:
\[
F(Z_N) \ge F(A^*) \ge \frac{F(S) + F(S')}{2}
\]

Setting $S'' = Z_N$ fulfills all constraints, concluding the proof.
\end{answer}

\setcounter{secnumdepth}{1}
\section{Human--AI Interaction Log II: Arriving at Lagrangian Relaxation}\label{app:ai-lagrangian}

\newcommand{\abbrivplain}{{\textcolor{darkgray}{[...] (ABBREVIATED)}}}
\newcommand{\abbriv}{{\vspace{10pt}\textcolor{darkgray}{[...] (ABBREVIATED)}}\vspace{5pt}}

\setcounter{secnumdepth}{0}

In this section, we include the sequence of interactions with Gemini Pro model that lead to the formulation of the Lagrangian relaxation approach together with the concavity theorem.
At the beginning, the model suggested a different (incorrect) algorithm using Maximum Flow, but after pointing out that its algorithm is incorrect, the model converges to the Lagrangian relaxation approach.
However, as noted earlier, the proof given by the model is incorrect. 
While we have attempted to make the model corrects its mistake, the proof remains wrong until the very end. 

Since the model's responses are quite long, we abbreviate part of the responses (that we believe are less relevant)---this is denoted by ``\abbrivplain'' in the responses below.

We also note that there is in fact an error in the problem formulation in the very first user prompt where $\sum_{j \in [m]} \sum_{i \in [w_{c(j)}]} w_i$ should instead be $\sum_{j \in [m]} \sum_{i \in [c(j)]} w_i$. This is a genuine mistake on our part. Despite this typo, the model immediately starts working on the correct version of the problem.

{

\setcounter{secnumdepth}{0}

\begin{prompt}
Consider the following problem: We have intervals $[l_1, r_1], \dots, [l_n, r_n]$ where $l_i, r_i$ are integers in $\{0, \dots, m\}$. Furthermore, we also have scores $w_1 \geq \dots \geq w_k \geq 0$. The goal is to pick $k$ intervals that maximizes the following objective $\sum_{j \in [m]} \sum_{i \in [w_{c(j)}]} w_i$
where $c(j)$ denotes the number of times $j$ is covered by the picked intervals. (Note that $c(j)$ is always a number from $0, \dots, k$.)

Can you prove that this problem is NP-hard?
\end{prompt}

\begin{answer}
Based on the problem statement provided, the problem is \textbf{not NP-hard}. In fact, it can be solved in \textbf{polynomial time}.

The problem describes a variation of the \textbf{Maximum Coverage Problem} specialized for \textbf{Intervals}, with an objective function characterized by \textbf{diminishing returns} (concavity). While the general Maximum Coverage problem (on arbitrary sets) is NP-hard, the interval structure and the concavity of the weights allow for efficient solutions using \textbf{Min-Cost Flow} or \textbf{Linear Programming}.

\abbriv
\end{answer}

\begin{prompt}
Ok, this is good. Can you please help write a full algorithm together with a full formal proof of correctness? You may assume that $m$ is not too large and just create a graph on the $m$ time steps directly without identifying the unique endpoints first.
\end{prompt}

\begin{answer}
Here is the full algorithm and a formal proof of correctness.

\paragraph{Problem Restatement.}

We are given $n$ intervals $[l_1, r_1], \dots, [l_n, r_n]$ within the range $\{0, \dots, m\}$. We have weights $w_1 \geq \dots \geq w_k \geq 0$. We must select a subset of indices $S \subseteq \{1, \dots, n\}$ with $|S| = k$ to maximize:

\[
\sum_{j \in [m]} \sum_{i=1}^{c(j)} w_i
\]

where $c(j)$ is the number of selected intervals covering point $j$.

\medskip
\hrule
\medskip

\paragraph{Part 1: The Algorithm (Min-Cost Flow Construction).}

We reduce this problem to a \textbf{Min-Cost Max-Flow} problem on a directed acyclic graph.

\paragraph{1. Graph Construction.}

Let $G = (V, E)$ be a directed graph.

\begin{itemize}
\item \textbf{Nodes:} Create $m+2$ nodes labeled $0, 1, \dots, m, m+1$.
\item Node $0$ acts as the \textbf{Source ($S$)}.
\item Node $m+1$ acts as the \textbf{Sink ($T$)}.
\item The nodes $1, \dots, m$ represent the time steps.
\item \textbf{Edges:} We construct two types of edges: \textbf{Interval Edges} and \textbf{Backbone (Coverage) Edges}.
\end{itemize}

\textbf{A. Interval Edges (Representing choice of intervals).}

For each interval $I_u = [l_u, r_u]$ in the input (where $u \in \{1, \dots, n\}$):
\begin{itemize}
\item Add a directed edge from $l_u$ to $r_u + 1$.
\item \textbf{Capacity:} $1$
\item \textbf{Cost:} $0$
\item \textit{Intuition:} Flowing through this edge represents ``picking'' interval $u$. It skips the ``costly'' backbone edges between $l_u$ and $r_u+1$.
\end{itemize}

\textbf{B. Backbone Edges (Representing coverage count).}

For each position $j \in \{0, \dots, m\}$, add $k$ parallel edges from node $j$ to node $j+1$.
Let's index these parallel edges by $t \in \{1, \dots, k\}$.
\begin{itemize}
\item \textbf{Capacity:} $1$ (per parallel edge)
\item \textbf{Cost:} $-w_t$
\item \textit{Intuition:} These edges represent the benefit of covering point $j$. If 1 unit of flow goes through the backbone at $j$, we gain $w_1$ (cost $-w_1$). If 2 units go through, we gain $w_1 + w_2$ (cost $-w_1 - w_2$).
\end{itemize}

\paragraph{2. The Flow Requirement.}

We require a total flow of exactly \textbf{$k$} units from Source ($0$) to Sink ($m+1$).

\paragraph{3. Execution.}

\begin{enumerate}
\item Construct the graph as described.
\item Run a standard \textbf{Min-Cost Max-Flow} algorithm (e.g., successive shortest path using Bellman-Ford or SPFA if weights are negative, or use potentials with Dijkstra) to send exactly $k$ units of flow from $0$ to $m+1$.
\item The minimum cost found will be negative. The maximum objective value for the original problem is \textbf{$-(MinCost)$}.
\end{enumerate}

\paragraph{4. Extracting the Solution.}

To identify which intervals to pick:

\begin{itemize}
\item Inspect the \textbf{Interval Edges}. If the edge corresponding to interval $[l_u, r_u]$ carries flow $1$, include interval $u$ in your solution set $S$.
\end{itemize}

\abbriv
\end{answer}

\begin{prompt}
Your algorithm is obviously incorrect since it allows more than $k$ intervals to be picked. Please fix this.

If you cannot fix it, try proving that it is NP-hard instead.
\end{prompt}

\begin{answer}
You are absolutely correct. My previous construction was flawed because it failed to strictly enforce the ``exactly $k$ intervals'' constraint (it allowed ``stitching'' disjoint intervals into a single unit of flow).

The problem is \textbf{not NP-hard}. It is solvable in \textbf{polynomial time} ($O(Poly(n, m) \cdot \log (\text{Precision}))$).

The correct approach requires handling the ``exactly $k$'' constraint using \textbf{Lagrangian Relaxation} (often called the WQS Binary Search or ``Alien's Trick'' in competitive programming) coupled with a \textbf{Min-Cost Circulation} algorithm.

Here is the corrected algorithm and the logic behind it.

\paragraph{The Core Logic: Lagrangian Relaxation.}

The constraint ``pick exactly $k$ intervals'' is hard to model directly in flow networks because ``count'' is a global property. However, we can move this constraint into the objective function.

Instead of fixing $k$, let's assign a \textbf{cost $\lambda$} to the act of picking \textit{any} interval.
We then solve the simpler unconstrained problem:

\begin{quote}
\textbf{Subproblem:} Pick \textit{any} set of intervals to maximize:
\[
\left( \sum_{j} \sum_{i=1}^{c(j)} w_i \right) - \lambda \cdot (\text{number of picked intervals})
\]
\end{quote}

Intuitively:

\begin{itemize}
\item If $\lambda$ is very high (expensive), the optimal solution picks few intervals.
\item If $\lambda$ is low (cheap), the optimal solution picks many intervals.
\item Since the marginal gain of adding intervals is diminishing (concavity), the number of intervals picked is monotonic with respect to $\lambda$. We can \textbf{binary search} for the $\lambda$ that results in picking exactly $k$ intervals.
\end{itemize}

The function $G(k) = \text{Optimal Score with } k \text{ intervals}$ is concave. This allows the binary search to find the optimal solution effectively.

\medskip
\hrule
\medskip

\subsection{The Algorithm}

\paragraph{1. Binary Search for $\lambda$.}

We binary search for a penalty value $\lambda \in [-M, M]$ (where $M$ is a sufficiently large number).
In each step of the binary search, we solve the \textbf{Min-Cost Circulation} problem described below.

\paragraph{2. Network Construction (Min-Cost Circulation).}

For a fixed $\lambda$, we construct a graph to find the optimal subset of intervals.

\begin{itemize}
\item \textbf{Nodes:} Create nodes $0, 1, \dots, m, m+1$ representing the timeline points.
\item \textbf{Timeline Edges (Left-to-Right):} These represent the ``state'' of being covered.
\item For each $j \in \{0, \dots, m\}$, add edges from $j \to j+1$.
\item Since the benefit is concave ($w_1 \ge w_2 \dots$), we model this with \textbf{Convex Cost} flow.
\item Add $n$ parallel edges from $j \to j+1$:
\begin{itemize}
\item Edge 1: Capacity 1, Cost $-w_1$.
\item Edge 2: Capacity 1, Cost $-w_2$.
\item \dots and so on.
\end{itemize}
\item \textit{Note:} Maximizing positive weight is equivalent to minimizing negative cost.
\item \textbf{Interval Edges (Right-to-Left):} These represent the decision to ``pick'' an interval.
\item For each available interval $i$ given as $[l_i, r_i]$, add a directed edge from \textbf{$(r_i + 1) \to l_i$}.
\item \textbf{Capacity:} 1 (each interval can be picked at most once).
\item \textbf{Cost:} $\lambda$.
\end{itemize}

\paragraph{3. Solving the Circulation.}

Run a \textbf{Min-Cost Circulation} algorithm. This algorithm finds a flow that satisfies flow conservation at every node and minimizes total cost.

\begin{itemize}
\item A \textbf{circulation} decomposes into a set of cycles.
\item A cycle here corresponds to:
\begin{enumerate}
\item Using an interval edge $(r_i+1) \to l_i$ (Paying $\lambda$).
\item Flowing back through the timeline $l_i \to \dots \to r_i+1$ (Collecting benefits $-w_1, -w_2 \dots$).
\end{enumerate}
\item The algorithm will only saturate a cycle if the benefit (sum of $w$'s) exceeds the cost $\lambda$.
\end{itemize}

\paragraph{4. Update Step.}

Let $k'$ be the number of interval edges saturated (carrying flow) in the optimal circulation.

\begin{itemize}
\item If $k' > k$: The penalty $\lambda$ is too low. Increase $\lambda$.
\item If $k' < k$: The penalty $\lambda$ is too high. Decrease $\lambda$.
\item If $k' = k$: We found the optimal set.
\end{itemize}

The final answer is:

\[
\text{Objective} = -(\text{MinCost from Circulation}) + k \cdot \lambda
\]

\medskip
\hrule
\medskip

\paragraph{Proof of Correctness (Sketch).}

\begin{enumerate}
\item \textbf{Total Unimodularity (TU):} The constraint matrix of the interval problem (an interval matrix) is Totally Unimodular. This ensures that the Linear Programming relaxation of the problem has an integral optimal solution.
\item \textbf{Concavity:} The function $f(x) = \text{max coverage with } x \text{ intervals}$ is concave (discrete concave) due to the diminishing returns of overlaps ($w_i \ge w_{i+1}$) and the TU property.
\item \textbf{WQS Binary Search Validity:} The concavity ensures that for any $k$, there exists a ``slope'' $\lambda$ such that the unconstrained optimization with penalty $\lambda$ yields an optimal solution with exactly $k$ items (or allows us to interpolate between $k$ and $k+1$).
\item \textbf{No Stitching:} The ``stitching'' issue of the previous algorithm is solved by the topology. Flow travels $L \to R$ on the timeline and $R \to L$ on interval edges. To reuse the timeline flow (stitching), one would need to jump from the end of one interval to the start of another \textit{against} the flow direction, which is impossible without using another costly interval edge or traversing the timeline (which correctly accounts for coverage).
\end{enumerate}

Thus, the problem is solvable in polynomial time using Min-Cost Circulation within a binary search.
\end{answer}

\begin{prompt}
Can you give some more justification regarding concavity? I actually don't see how it holds.
\end{prompt}

\begin{answer}
This is an excellent question. The claim of concavity is the cornerstone of the polynomial-time solution (specifically for the WQS binary search), so it requires a rigorous proof rather than just an appeal to ``diminishing returns.''

The concavity holds because the problem can be transformed into a \textbf{Min-Cost Flow problem with a parameter $k$ as the flow demand}. A fundamental property of Min-Cost Flow networks is that the minimum cost is a \textbf{convex function} of the flow amount.

\abbriv
\end{answer}

\begin{prompt}
Your proof is wrong: We already discussed this but the minimum cost flow formulation you used in the proof is incorrect because it does *not* take the cardinality constraint into account.

Can you just try to prove this from first principle? Namely, just try to argue, if we have two solutions for $k - 1$ and $k + 1$, how do we combine them to get the solution for $k$? Again, be completely formal and self-contained.
\end{prompt}

\begin{answer}
You are asking for a direct, first-principles justification for why the algorithm works. The core of the correctness relies on the \textbf{concavity} of the optimal objective function $F(k)$ (the max score with $k$ intervals).

If $F(k)$ is concave (i.e., $2F(k) \ge F(k-1) + F(k+1)$), then the local marginal gains are diminishing, and the \textbf{WQS Binary Search} (Lagrangian Relaxation) is guaranteed to find the global optimum.

Below is the formal proof of concavity. It avoids flow networks and relies on the geometric properties of intervals---specifically, the \textbf{Total Unimodularity (TU)} of the interval matrix.

\subsection{Problem Setup}

Let $\mathcal{I}$ be a collection of $n$ intervals. Let $A$ be the $m \times n$ coverage matrix, where $A_{ji} = 1$ if interval $i$ covers point $j$, and $0$ otherwise.
Let $w: \mathbb{Z}_{\ge 0} \to \mathbb{R}$ be a function defining the score of covering a point $c$ times:
\[ w(c) = \sum_{t=1}^{\min(c, k)} w_t \]
Since $w_1 \ge w_2 \ge \dots$, the function $w(c)$ is \textbf{concave}.
The total objective for a selection vector $x \in \{0, 1\}^n$ is:
\[ \text{Obj}(x) = \sum_{j=1}^m w( (Ax)_j ) \]
Let $F(k)$ be the maximum objective value subject to $\sum x_i = k$.

\paragraph{The Claim.}

We prove that $F(k)$ is concave. Specifically, for any $k$:
\[ F(k-1) + F(k+1) \le 2F(k) \]

\subsection{The Proof}

\paragraph{Step 1: Combining Solutions.}
Let $x$ be an optimal solution for $k-1$ (so $\sum x_i = k-1$).
Let $z$ be an optimal solution for $k+1$ (so $\sum z_i = k+1$).

Consider the sum vector $h = x + z$. Note that $h \in \{0, 1, 2\}^n$.
The total number of intervals in $h$ is $(k-1) + (k+1) = 2k$.
The coverage depth at point $j$ for this combined multiset is $C_j = (Ax)_j + (Az)_j$.

We wish to decompose $h$ into two new integer vectors $y_1, y_2 \in \{0, 1\}^n$ such that:
\begin{enumerate}
    \item $y_1 + y_2 = h$ (They partition the multiset $h$).
    \item $\sum y_{1,i} = k$ and $\sum y_{2,i} = k$ (Both are feasible for problem size $k$).
    \item Their coverage profiles are "balanced" to exploit the concavity of $w$.
\end{enumerate}

\paragraph{Step 2: The Splitting Lemma (via Total Unimodularity).}

\begin{lemma}
Let $M$ be a matrix that is Totally Unimodular (TU). For any integer vector $h$, there exist integer vectors $y_1, y_2$ such that $y_1 + y_2 = h$ and for every row $r$:
\[ |(My_1)_r - (My_2)_r| \le 1 \]
\end{lemma}

\textbf{Application to our problem:}
The constraint matrix describing our problem consists of:
1. The interval coverage matrix $A$ (which has consecutive ones in columns).
2. The cardinality constraint (a row of all 1s).
It is a known result that an interval matrix augmented with a row of 1s is \textbf{Totally Unimodular}.

Applying the lemma to $h = x + z$:
We can find $y_1, y_2$ such that $y_1 + y_2 = h$, and:
\begin{itemize}
    \item \textbf{Cardinality Balance:} $|\sum y_{1,i} - \sum y_{2,i}| \le 1$. Since the total sum is $2k$, the sums must be exactly $k$ and $k$. Thus, $y_1, y_2$ are valid feasible solutions for the size $k$ problem.
    \item \textbf{Coverage Balance:} For every point $j$, the coverages $c_1 = (Ay_1)_j$ and $c_2 = (Ay_2)_j$ satisfy $|c_1 - c_2| \le 1$.
\end{itemize}

Since $c_1 + c_2 = C_j$ (total coverage at $j$), the condition $|c_1 - c_2| \le 1$ implies that $c_1$ and $c_2$ are the floor and ceiling of the average:
\[ \{c_1, c_2\} = \left\{ \left\lfloor \frac{C_j}{2} \right\rfloor, \left\lceil \frac{C_j}{2} \right\rceil \right\} \]

\paragraph{Step 3: Discrete Concavity Argument.}
We now compare the total scores.
We know that for any concave function $w$ and any integers $a, b$:
\[ w(a) + w(b) \le w\left(\left\lfloor \frac{a+b}{2} \right\rfloor\right) + w\left(\left\lceil \frac{a+b}{2} \right\rceil\right) \]
(Intuitively: pushing values closer to the mean increases the sum for concave functions).

Applying this to every point $j$:
\[ w((Ax)_j) + w((Az)_j) \le w((Ay_1)_j) + w((Ay_2)_j) \]

Summing over all points $j$:
\[ \text{Obj}(x) + \text{Obj}(z) \le \text{Obj}(y_1) + \text{Obj}(y_2) \]

Since $x$ and $z$ were optimal for $k-1$ and $k+1$:
\[ F(k-1) + F(k+1) \le \text{Obj}(y_1) + \text{Obj}(y_2) \]

Since $y_1$ and $y_2$ are valid solutions for $k$, their scores cannot exceed optimal $F(k)$:
\[ \text{Obj}(y_1) \le F(k) \quad \text{and} \quad \text{Obj}(y_2) \le F(k) \]

Combining these inequalities:
\[ F(k-1) + F(k+1) \le 2 F(k) \]

\paragraph{Conclusion.}
This proves that the objective function $F(k)$ is concave. Consequently, the WQS Binary Search algorithm (Lagrangian Relaxation) is mathematically justified and will correctly find the global maximum in polynomial time.
\end{answer}

\begin{prompt}
You claim that an interval matrix augmented with a row of 1s is Totally Unimodular. This is clearly wrong as can be seen by the following matrix [[1, 0, 0, 1], [0, 1, 0, 1], [0, 0, 1, 1], [1, 1, 1, 1]] which has determinant 2. Can you fix this issue?
\end{prompt}

\begin{answer}
You are absolutely correct. The matrix you provided is a valid counterexample. The ``interval matrix augmented with an all-ones row'' is \textbf{not} totally unimodular (TU) in general. I apologize for that error.

While the specific matrix formulation involving the cardinality row is not TU, the \textbf{Maximum Weight $k$-Interval Coverage problem is still solvable in polynomial time}.

The correctness of the algorithm relies on the \textbf{concavity of the objective function} $F(k)$. We can prove this concavity via a reduction to \textbf{Min-Cost Flow}, relying on the fact that the \textbf{node-arc incidence matrix} of a network is always TU (unlike the interval coverage matrix with a cardinality constraint).

Here is the corrected formal proof that resolves the issue by shifting the perspective from generic linear programming to network flow theory.

\paragraph{Addressing the Matrix Issue.}
The user correctly identified that the constraint matrix formed by interval coverage constraints plus a global cardinality constraint is not necessarily Totally Unimodular (TU). Therefore, we cannot rely on the standard LP relaxation of that specific matrix formulation.

However, we can model the problem exactly using a \textbf{Network Flow} formulation. The constraint matrix of a network flow problem (the node-arc incidence matrix) is always TU. The cardinality constraint $k$ is handled as the \textbf{flow demand}, which acts as the right-hand side of the linear system rather than a row in the matrix itself. This preserves the integrality of the solution.

\subsection{Formal Proof}

\begin{theorem}
Let $F(k)$ be the maximum score achievable using exactly $k$ intervals. The function $F(k)$ is concave.
\end{theorem}

\begin{proof}
We construct a Min-Cost Flow network $G=(V, E)$ such that the minimum cost of sending $k$ units of flow is exactly $-F(k) + C_{constant}$. Since the minimum cost of a flow network with convex edge costs is a convex function of the flow amount, this implies $F(k)$ is concave.

\paragraph{Network Construction.}
Create a graph with nodes $0, 1, \dots, m, m+1$.
\begin{itemize}
    \item \textbf{Source/Sink:} We inject $k$ units of flow at node $0$ and extract $k$ units at node $m+1$.
    \item \textbf{Interval Arcs (Selection):} For every interval $i$ given as $[l_i, r_i]$, create a directed edge $u_i$ from $l_i$ to $r_i+1$.
    \begin{itemize}
        \item Capacity: $1$.
        \item Cost: $0$.
        \item Interpretation: Flow $x_i=1$ on this arc means interval $i$ is selected.
    \end{itemize}
    \item \textbf{Timeline Arcs (Non-coverage):} For every $j \in \{0, \dots, m\}$, create a directed edge $v_j$ from $j$ to $j+1$.
    \begin{itemize}
        \item Capacity: $k$.
        \item Cost Function $\Phi_j(y)$: Depends on flow $y$ through this arc.
    \end{itemize}
\end{itemize}

\paragraph{Flow Conservation and Costs.}
Let $y_j$ be the flow on the timeline arc $(j, j+1)$.
Let $S_j$ be the set of intervals covering point $j$. In the network, these correspond to interval arcs that "jump over" the timeline segment $(j, j+1)$.
Consider the cut separating $\{0, \dots, j\}$ from $\{j+1, \dots, m+1\}$. The total flow across this cut must be $k$. The edges crossing this cut are exactly the timeline arc $v_j$ and the interval arcs covering $j$.
\[ y_j + \sum_{i \in S_j} x_i = k \]
Thus, the number of active intervals covering $j$ is $c(j) = \sum x_i = k - y_j$.
The flow $y_j$ on the timeline represents the "missed" coverage opportunities (or the number of empty slots).

We want to maximize $\sum_{t=1}^{c(j)} w_t$.
In terms of $y_j$, we maximize $\sum_{t=1}^{k-y_j} w_t$.
This is equivalent to minimizing the negative score. We define the cost function $\Phi_j(y_j)$ for the timeline arc as:
\[ \Phi_j(y_j) = - \sum_{t=1}^{k-y_j} w_t \]

\paragraph{Convexity of Edge Costs.}
To use the properties of Min-Cost Flow, the edge cost functions $\Phi_j(y)$ must be convex.
Let us examine the marginal cost $\Delta \Phi_j(y) = \Phi_j(y+1) - \Phi_j(y)$.
\[ \Phi_j(y) = - (w_1 + \dots + w_{k-y}) \]
\[ \Phi_j(y+1) = - (w_1 + \dots + w_{k-y-1}) \]
\[ \Delta \Phi_j(y) = \Phi_j(y+1) - \Phi_j(y) = w_{k-y} \]
The sequence of marginal costs for $y=0, 1, \dots$ is $w_k, w_{k-1}, \dots, w_1$.
Since the problem specifies non-increasing weights $w_1 \ge w_2 \ge \dots \ge w_k$, the sequence of marginal costs is \textbf{non-decreasing}.
A discrete function with non-decreasing marginal differences is \textbf{convex}.

\subsection{Conclusion.}
\begin{enumerate}
  \item The constraint matrix of this problem is the node-arc incidence matrix of the graph $G$, which is \textbf{Totally Unimodular}. Thus, for integer demand $k$, there exists an integral optimal flow.
  \item A standard result in network flow theory states that the optimal objective value of a Min-Cost Flow problem with convex edge costs is a \textbf{convex function} of the flow amount $k$.
  \item Let $C(k)$ be this minimum cost. Our original maximization objective is $F(k) = -C(k)$.
  \item Since $C(k)$ is convex, $F(k)$ is \textbf{concave}.
\end{enumerate}

This validates the WQS Binary Search approach used in the algorithm.
\end{proof}

\abbriv
\end{answer}

\setcounter{secnumdepth}{1}
\fi

\end{document}